\documentclass[preprint,9pt]{sigplanconf}

\usepackage{amsmath}
\usepackage{epsfig}
\usepackage{algorithmic}
\usepackage{algorithm}
\usepackage{url}
\usepackage{ifthen}
\newboolean{ThisIsTR}
\setboolean{ThisIsTR}{true}

\newtheorem{theorem}{Theorem}[section]
\newtheorem{lemma}[theorem]{Lemma}

\newenvironment{proof}[1][Proof]{\begin{trivlist}
  \item[\hskip \labelsep {\bfseries #1}]}{\end{trivlist}}

\newcommand{\qed}{\nobreak \ifvmode \relax \else
         \ifdim\lastskip<1.5em \hskip-\lastskip
         \hskip1.5em plus0em minus0.5em \fi \nobreak
         \vrule height0.75em width0.5em depth0.25em\fi}

\begin{document}



\title{Variable and Thread Bounding for Systematic Testing of Multithreaded Programs}

\authorinfo{Sandeep Bindal\and Sorav Bansal}
           {IIT Delhi}
           {\{cs5080536, sbansal\}@cse.iitd.ernet.in}
\authorinfo{Akash Lal}
           {Microsoft Research}
           {akashl@microsoft.com}

\maketitle

\begin{abstract}
  Previous approaches to systematic state-space exploration for testing multi-threaded
  programs have proposed context-bounding~\cite{musuvathi:icb:pldi07} and
  depth-bounding~\cite{pct} to be effective ranking algorithms
  for testing multithreaded programs. This paper proposes two new metrics
  to rank thread schedules for systematic state-space exploration. Our metrics
  are based on characterization of a concurrency bug using $v$ (the minimum number of
  distinct variables that need to be involved for the bug to manifest) and
  $t$ (the minimum number of distinct threads among which scheduling constraints are
  required to manifest the bug).
  Our algorithm is based on the hypothesis
  that in practice, most concurrency bugs have low $v$ (typically 1-2)
  and low $t$ (typically 2-4) characteristics.
  We iteratively explore the search space of schedules 
  in increasing orders of $v$ and $t$. We show qualitatively and empirically that
  our algorithm finds common bugs in fewer
  number of execution runs, compared with previous approaches. We also show
  that using $v$ and $t$
  improves the lower bounds on the probability of finding bugs
  through randomized algorithms.
  
  Systematic exploration of schedules requires instrumenting each
  variable access made by a program, which can be very expensive
  and severely limits the applicability of this
  approach. Previous work \cite{musuvathi:icb:pldi07, pct} has
  avoided this problem by interposing only on
  synchronization operations (and ignoring other variable accesses). We demonstrate
  that by using variable bounding ($v$) and a
  static imprecise alias analysis, we can interpose on all variable
  accesses (and not just synchronization operations) at 10-100x less
  overhead than previous approaches.
\end{abstract}

\category{D.2.4}{Software Engineering}{Software/Program Verification --- formal methods, validation}
\category{F.3.1}{Logics and Meanings of Programs}{Specifying and Verifying and Reasoning about Programs --- mechanical verification, specification techniques}
\category{D.2.5}{Software Engineering}{Testing and Debugging --- debugging aids, diagnostics, monitors, tracing}

\terms
Algorithms, Reliability, Verification

\keywords
Concurrency, context-bounding, variable-bounding, thread-bounding, model checking, multi-threading, concurrency-bug classification, shared-memory programs, software testing

\section{Introduction}
\label{sec:intro}
Testing concurrent programs is notoriously
difficult because of its inherent non-determinism. An
effective but expensive approach is {\em model-checking}, where all
possible schedules of a program are executed to ascertain the absence
of a bug.
Unfortunately, the space
of all schedules is huge, and exhaustively enumerating it
is usually infeasible. For a multi-threaded program with $n$ threads,
each executing $k$ instructions, the total number
of schedules (or thread interleavings)
is $\frac{(nk)!}{(k!)^n}$. This space of schedules further
explodes if each instruction is not guaranteed to be
atomic.
For a very small program with $k = 100$ and $n = 2$, the total
number of interleavings is around $10^{59}$!

As it is practically impossible to exhaustively explore the entire state
space of all schedules for any useful program, an alternative is to try
and maximize the probability of uncovering a bug rather than trying
to ascertain its absence.
Many different approaches have
been proposed in this direction. Musuvathi and Qadeer proposed using
{\em context-bound} to rank schedules, and show that it is an effective
method to uncover most common bugs \cite{musuvathi:icb:pldi07}.
A context-bound is the number of pre-emptive context-switches required
to execute a schedule. The schedules are enumerated in increasing
order of their context-bound, i.e., all schedules with context bound $c-1$ are
executed before any schedule with context bound $c$.
Musuvathi and Qadeer report experiments on real-world applications, and
show that all known bugs in those applications were found at context-bound
values of 2 or less.

Iterative context bounding is an effective way of ranking
schedules. However, this metric is often too coarse-grained.
For a
multi-threaded program with $n$ threads, each executing $k$ instructions, the
total number of schedules at context-bound $c$ grows with
$(nk)^{c}$.
For a small program with $k=10,000$ instructions and $n=4$, the number
of schedules at context bound $2$ is on the order of $10^{9}$!
Musuvathi et. al's concurrency-testing tool based on this algorithm, CHESS, reduces
this search space
by considering only explicit synchronization operations
as possible pre-emption points, thus reducing $k$ by at least 2-3 orders
of magnitude. This simplification is justified by the assumption
that most programs follow a mutual-exclusion locking discipline, and
hence all shared-memory accesses will be protected by {\tt lock()}
and {\tt unlock()} calls. Violation of this locking discipline can
be separately checked using other race-detection tools.
This approach, though effective, is not completely general, as
many systems deliberately avoid explicit
synchronization \cite{xiong:adhoc:osdi10}, often for performance reasons.

Another approach to testing multithreaded programs is randomization of scheduling
decisions with probabilistic guarantees. Burckhardt et. al. \cite{pct}
characterize a concurrency bug by its \emph{depth}---the minimum number of scheduling
constraints required to find the bug. They provide an algorithm that provides a lower bound
on the probability of finding a depth-$d$ bug. Ranking on bug-depth $d$ restricts
the search space of a multi-threaded program with $n$ threads and executing
$k$ instructions to $nk^{d-1}$. This, again, may be too large for most programs.

Another recent tool, CTrigger \cite{ctrigger}, focuses on atomicity-violation bugs and preferentially
searches the space of schedules that are likely to trigger these bugs.
CTrigger first profiles executions of the program to determine
the shared variables and their unprotected accesses. It then attempts
to generate schedules that are likely to violate assumptions
of atomicity (for example, by inserting a write to location $M$
by some thread between two accesses to the same location $M$ by another
thread).
CTrigger is primarily interested in atomicity-violation bugs and often
overlooks other concurrency bugs.

Our first contribution is to propose the use of number of variables to further
classify and reduce the schedule search space. Our algorithm is based on the
hypothesis that in practice, most concurrency bugs can be uncovered by restricting
our search to only a few variables at a time. At a time, we only search for bugs
involving a small subset of $v$ variables. These variables may include synchronization
operations. Iteratively, we consider all such variable subsets.
For a given subset of variables, we perform static alias analysis
to identify all program locations where these variables may be accessed. We
instrument only these program locations. This selective instrumentation allows
us to run our program at near-native speed. Consequently, our approach can interpose
on any variable accesses, and not just synchronization variables as reported in previous
work. We show that using variable bounding, the search space reduces by
a factor of roughly $(\frac{Q}{v})^{c-v}$ when searching for bugs with
context-bound
$c$ and variable bound $v$, where $Q$ is the total number of variables in
the program. We confirm this result
experimentally by showing that variable bounding allows faster discovery
of concurrency bugs.

Our second contribution is characterizing a concurrency bug by
the number of distinct threads that need to be order-constrained to uncover the bug.
A bug that can be uncovered by constraining the order of $t$ threads is called
a $t$-thread bug.
In practice, most bugs have a small $t$.
We provide a randomized algorithm with
guarantees on the probability of uncovering a $t$-thread bug, if it exists.
Using thread-bounding, the search space decreases by a factor
of $\frac{n!}{(t+1)!log(n)}$ when searching for bugs with thread-bound $t$ out
of a total of $n$ program threads.

We note that our hypothesis that most bugs can be uncovered at
low $(v,t)$ values conform
with the observations made in previous work on studying real-world
concurrency bug characteristics \cite{lu:learning_from_mistakes:asplos08}.

The paper is organized as follows. Section~\ref{sec:varbound} presents
and analyzes variable bounding for exhaustive model-checking algorithms.
Section~\ref{sec:varbound-random} discusses variable bounding for randomized
algorithms and analyzes the resulting probabilistic guarantees of finding a bug,
if one exists. Section~\ref{sec:threadbound} discusses thread bounding.
Sections~\ref{sec:implementation}~and~\ref{sec:results} discuss our implementation
and empirical results. Section~\ref{sec:relwork} discusses related work,
and Section~\ref{sec:conclusion} concludes.

\section{Variable Bounding}
\label{sec:varbound}
Recent work on studying characteristics of real-world
concurrency bugs~\cite{lu:learning_from_mistakes:asplos08} concluded that 66\% of
the non-deadlock concurrency bugs they examined involved only one variable.
Perhaps, the most common type of concurrency bug involving one variable
access is a data race. i.e., simultaneous access of a shared
variable (of which, one is a write) by two or more threads without proper
synchronization.
Also, among the remaining fraction of
non-deadlock concurrency bugs, most bugs involve only a few
variables (typically 2 to 3).
This observation
motivates our ranking on the number of memory locations involved. We
first enumerate schedules
that exhaustively check all thread interactions involving a single
variable. We then enumerate schedules that exhaustively check thread
interactions involving two variables, and so on.

We first discuss variable bounding in the context of a model-checker.
For a model-checker like CHESS~\cite{chess}, a custom priority scheduler
implements the exhaustive enumeration of schedules, and
context-bounding~\cite{musuvathi:icb:pldi07} is used to limit the number of
schedules executed.
To implement variable bounding, we first identify all program
variables (or points in the program that generate new variables) by parsing
the program.
These program variables include globals and heap-allocated
variables (allocated using {\tt malloc()} or {\tt new}).
A heap variable is identified and named
by its allocation statement and the number of times that
statement has been invoked.
For example, if a particular {\tt new} statement is called multiple
times, we will consider each return value as a separate variable.
We call this set of program variables $\vartheta$. Iteratively, we take
all $v$-sized subsets of variables in $\vartheta$
for $v \in \{1, 2, 3, \ldots\}$. For a subset $V$ of size $v$, we
execute schedules that explore {\em all} interactions between all variables
in $V$.

To identify variables, we instrument heap allocation statements to
generate a new variable name for each invocation of the statement.
As we explain later, we also
prioritize the variables which are generated in the first few loop
iterations.
To identify interactions between a subset of variables, we instrument
accesses to these variables. We use a lightweight and
imprecise static alias
analysis~\cite{andersen94programanalysis, whaley:alias:pldi04,
lhotak:alias:tosem08} to identify
program points
at which each variable in $\vartheta$ {\em may} be accessed. Our
static analysis
assumes that the program is {\em memory-safe}. i.e., locations outside
allocation boundaries will not be accessed. Memory-safety can be separately
checked using other available tools.

Without variable bounding, all accesses to all variables must be instrumented
with a call to the scheduler which implements exhaustive schedule enumeration.
With variable bounding, this instrumentation can be significantly reduced.
For
a variable $x_i \in \vartheta$, we call the set of program locations at which
it may be accessed $a_{x_i}$. With variable bounding, we only check
interactions within a variable subset
$V=\{x_0,x_1,\dots,x_v\}$, and instrument all
locations in the
set $(a_{x_0} \cup a_{x_1} \cup \dots \cup a_{x_v})$.
The instrumentation
code includes a call to a scheduler function, {\tt varaccess()} that
yields to the scheduler which
implements priority scheduling and systematic pre-emption.
{\tt varaccess()} is inserted after the program has accessed and possibly
updated the variable.
To ensure that pre-emption occurs only on accesses to the set of
tracked variables, the instrumentation code dynamically checks
that the accessed memory address is one of the tracked variables before calling
{\tt varaccess()}.
The {\tt varaccess()} call serves as a potential yield
point (or context-switch point), i.e.,
at this point, the scheduler can choose to run another thread.
To allow a thread to be pre-empted before its first access to a
variable, we also insert a {\em fake} {\tt varaccess()} before the first
instruction of each thread.
Our enumeration algorithm is similar to that used in CHESS~\cite{chess} and
we discuss it in Section~\ref{sec:implementation}.

\subsubsection*{Bug Characterization}
We call a concurrency bug a $c$ context bug if at least $c$ pre-emptive context
switches are required for the bug to manifest. $c$ is also called the bug's
context-bound. This definition of context bound is taken
from previous work \cite{musuvathi:icb:pldi07}.

We call a concurrency bug a $v$-variable bug if the minimal set of constraints
required to manifest the bug involve preemption points at accesses to $v$ distinct
variables. $v$ is also called the bug's variable bound. By definition, $v\leq c$
for any $c,v$ bug.

Figures~\ref{fig:e-c0v0},~\ref{fig:e-c1v1},~\ref{fig:e-c2v1},~\ref{fig:e-c2v2} show
short programs with $(c=0,v=0)$, $(c=1,v=1)$, $(c=2,v=1)$, $(c=2,v=2)$ bugs
respectively for exposition. The numbers in comments give the order of execution for
an assertion failure. In these short programs, we count a pre-emption
against the shared variable that was last accessed. Also, we assume that a bug
exists if the ASSERT statement can fail.
\begin{figure}[htb]
\begin{center}
\begin{tabular}{p{3cm}|@{\ \ \ \ }p{3cm}}
\multicolumn{2}{c}{a = 0}\\
Thread 1:            & Thread 2:\\
\ \ ASSERT(a == 0); // 2  & \ \ a++; // 1 \\
\end{tabular}
\caption{\label{fig:e-c0v0}A short program with a $c=0,v=0$ bug}
\end{center}
\end{figure}

\begin{figure}[htb]
\begin{center}
\begin{tabular}{p{3cm}|@{\ \ \ \ }p{3cm}}
\multicolumn{2}{c}{a = 0}\\
Thread 1:             &Thread 2:\\
\ \ t1 = a; // 1          &\ \ a++; // 2\\
\ \ t2 = a; // 3          &\\
\ \ ASSERT(t1==t2); // 4 &\\
\end{tabular}
\caption{\label{fig:e-c1v1}A short program with a $c=1,v=1$ bug}
\end{center}
\end{figure}

\begin{figure}[htb]
\begin{center}
\begin{tabular}{p{3cm}|@{\ \ \ \ }p{3cm}}
\multicolumn{2}{c}{a = 0}\\
Thread 1:                & Thread 2:\\
\ \ t1 = a; // 2             &\ \ a = 1; // 1\\
\ \ t2 = a; // 4             &\ \ a = 0; // 3\\
\ \ ASSERT(t1==t2); // 5\\
\end{tabular}
\caption{\label{fig:e-c2v1}A short program with a $c=2,v=1$ bug}
\end{center}
\end{figure}

\begin{figure}[htb]
\begin{center}
\begin{tabular}{p{4.2cm}|@{\ \ \ \ }p{3cm}}
\multicolumn{2}{c}{\ \ \ \ a = 0, b = 0}\\
Thread 1:                & Thread 2:\\
\ \ t1 = a; // 1             &\ \ a = 1; // 2\\
\ \ t2 = a; // 3             &\ \ b = 1; // 4\\
\ \ t3 = b; // 5              &\ \ b = 0; \\
\ \ ASSERT(t1==t2 or t3 != 1); // 6 & \\
\end{tabular}
\caption{\label{fig:e-c2v2}A short program with a $c=2,v=2$ bug}
\end{center}
\end{figure}

%
%
%
%
%

\subsubsection*{Schedule Characterization}
A schedule is characterized by $c$---the number of pre-emptive context switches in it,
and $v$---the number of distinct variables at which a pre-emptive context switch was performed. 

\subsubsection*{Search Space Reduction}
We now discuss how variable bounding helps reduce the search space.
Let us assume that a multi-threaded program with $t$ threads
has $Q$ distinct shared variables, represented as a set $\vartheta$ of variables,
i.e., $|\vartheta|=Q$. For simplicity, let us also assume that
each thread in the program accesses each variable in $\vartheta$ exactly
$d$ times. Hence, the total number of variable accesses by a thread
are $dQ$. Assuming that only accesses to these shared variables
are interesting context-switch points, and assuming $n$ threads,
$k=ndQ$ ($k$ is the number of
steps in a program).
Therefore, the number of schedules that need to be explored
at context bound $c$ are $O((ndQ)^{c})$. Let us call this expression $A$.

If we focus on a subset $V \subset \vartheta$ of $v$ variables, the
number of schedules that need to be explored at context bound $c$
are $\binom{Q}{v}(ndv)^{c}$ (first choose a subset $V \subset \vartheta$,
then explore all schedules with preemptions at accesses to variables in $V$).
Assuming $v,c \ll Q$, this expression is $O(Q^{v}(ndv)^c)$.
Comparing with $A$, we see that this expression is less than $A$
if $v < c$. This reduction in the search space (number of execution runs)
is significant for programs with a large number of variables (large $Q$). Apart
from this reduction in the number of execution runs, the time taken by each
execution run also decreases dramatically with variable bounding, as only the
accesses to variables being tracked need to be instrumented. We study both these improvements
in detail in our experiments in Section~\ref{sec:results}.

At $v = c$, variable bounding provides no improvement in the size
of the search space, but still provides a significant reduction in runtime because
of much lower instrumentation overhead (only the tracked variables need to
be instrumented). 
Effectively, by slicing the program into accesses to a small subset
of variables, we reduce the number of program steps $k$.
This is because only
accesses to the variable being tracked are considered valid context switch
points.
As we discuss in our experiments (Section~\ref{sec:results}), this reduction
is significant for most programs.
This method of reducing $k$ is more general than the approach
used in previous tools (e.g., CHESS~\cite{chess}) where all accesses to
non-synchronization variables
are ignored.

While we have used a simplified assumption of constant number of accesses $d$ to
each variable by each thread, the result does not change (although the analysis
gets more involved) if we assume varying number of accesses by each thread to different
variables. The same result can be obtained by replacing $d$ with the average
number of accesses by a thread to a randomly-chosen variable, and we skip this
discussion for brevity. We analyze a more general scenario in our discussion on
probability bounds for randomized bug-finding
algorithms (Section~\ref{sec:varbound-random}).

\subsection*{Heap Allocated Variables and Arrays}
\label{sec:heap_allocated_variables}
Our set of tracked variables include heap-allocated variables.
Heap-allocated variables are named using the heap-allocation statement and
the number of times that statement was executed before this variable was
generated. A large number of heap allocations by one statement
can generate a large number of variables causing our variable-bounding
algorithm to get stuck at low $v$ values.

In our experience, if the program contains a bug involving a certain {\em type} of
heap variable, the bug usually manifests while tracking the first few variables
of that type. For example, if the program constructs and accesses
a heap data structure (e.g., linked list), it
is very likely that a bug, if it exists, will be exposed by exploring all interactions
among the first few elements of that data structure.

The challenge is to identify and group variables of a certain type, so that
only the first few variables of that type are considered. We use a simple
heuristic that we found to work well in practice. The type of a variable is
defined by the callstack at the time of allocation of that variable. We expect
that largely, variables allocated with identical callstacks are of the same
type.
This heuristic is neither sound nor complete. For example, it is possible that
variables of the same type are allocated at different points in the
program, hence having different callstacks. This can cause our algorithm
to execute more than the required number of schedules. A more serious problem
is that
two identical callstacks could generate completely different types of
variables. This can cause our algorithm to overlook certain bugs.
Fortunately, in practice, such code is rare.

The algorithm works as follows. For each heap allocation, we generate a new
variable ID labeled by the location of the allocation statement and the number
of times that statement was executed. With each variable ID, we also associate
the number of times this allocation statement has previously been executed with
an identical callstack. We call this latter number, the loop iteration
number (because the allocations with identical callstacks must
be happening through a loop) of that variable. We first search for bugs
involving variables with lower loop iteration numbers before searching for bugs
involving variables with higher loop iteration numbers. We call this algorithm
loop-iteration bounding and denote the current loop-iteration number being
searched with letter $l$.
Figure~\ref{fig:instrumentation_code} shows our logic for implementing loop
iteration numbers. Note that, by design, variables allocated by recursive calls with
different recursion depth will be named differently (because they will have
different call stacks).

\begin{figure}[htb]
  {\tt
  <instrumentation code for new()>\\
  callstack := get\_current\_callstack();\\
  $v$ := <heap-allocation-statement, alloc\#>;\\
  lin := loop\_iteration\_number(callstack);\\
  increment\_loop\_iteration\_number(callstack);\\
  if (lin <=  $l$) \{\\
  \begin{tabular}{@{}p{0.2cm}@{}l}
  & add $v$ to the set $Q$ of the variables to be tracked;
  \end{tabular}
  \}\\
  }
  \caption{\label{fig:instrumentation_code}Instrumentation code for heap-allocation
  statements that considers only variables with loop-iteration number $\leq l$.}
\end{figure}

We also need special handling for array variables. Whenever possible, we treat
each location in the array as a separate variable. If the search space
size becomes unmanageably large (for high values of $v$), we use a less precise
but sound approach of considering the whole array as a single variable.

\section{Variable Bounding on Randomized Algorithms}
\label{sec:varbound-random}
Apart from exhaustive state space exploration to ascertain the absence
of certain bugs, randomized schedulers that provide probabilistic
guarantees of finding certain types of bugs have also been
proposed. Depth-bounding~\cite{pct} (also called Probabilistic Concurrency
Testing in the paper) is one such approach. The primary advantage of
randomized approaches over exhaustive search is that
the former can cover a large part of the program in relatively fewer
runs. Exhaustive search, on the other hand, can get stuck in local regions
of the program for long periods of time causing bugs in other regions
to go undetected. In this section, we discuss variable bounding
in the context of randomized search.

In particular, we study Probabilistic Context Bounding (PCT)~\cite{pct} that
proposed the \emph{bug-depth} metric.
While we analyze only PCT, similar arguments
will hold for other randomized algorithms.
For a program spawning at most $n$ threads and executing at most $k$ total instructions,
PCT algorithm works as follows (for an input parameter $d$, denoting the depth of the bug being
searched):
\begin{enumerate}
\item Assign $n$ priority values $d$, $d+1$,\ldots, $d+n$ randomly to the $n$ threads.
\item Pick $d-1$ priority change points $k_1$,\ldots, $k_{d-1}$ randomly in the range
$\lbrack1,k\rbrack$. Each $k_i$ has an associated priority value of $i$.
\item Schedule the threads by honoring their priorities, i.e., always execute an enabled thread
with the highest priority. When a thread reaches the $i$-th
change point (i.e., executes the $k_i$-th instruction), change the priority of that thread
to $i$.
\end{enumerate}
Burckhardt et. al.~\cite{pct} proved that this algorithm
finds a bug of depth $d$ with probability at least $1/nk^{d-1}$.

We implement variable bounding on PCT by first randomly choosing a set
of $v$ variables, and then
randomly choosing $d-1$ priority change points at one of the
accesses to the chosen variables (other instructions in the program are
not considered as potential priority change points).
For heap-allocated variables, we simply choose a heap allocation
statement ({\tt new} and {\tt malloc}) in lieu of a variable. Accesses
to any of the variables
allocated at the chosen heap-allocation statement are considered
potential priority
change points.

Notice that using a heap-allocation statement as one ``variable'' in the randomized
algorithm is a departure from the strategy used
in the exhaustive-search strategy, where each heap allocation is considered
a separate variable. This is done to ensure that we know the number of
these variables at compile time, and hence can appropriately choose a variable
set to provide probabilistic guarantees. Under this new definition of a variable, a
$v$-variable bug is a bug that involves memory locations allocated at at most
$v$ distinct heap-allocation statements (or globals). This new definition
performs a coarser classification of program's memory locations.
This could potentially cause higher number of required executions for
effective state space search for the same $v$ value. However, this is still a
significant improvement over not using variable bounding at all.
Also, this definition of variable
bounding does not make our argument on most bugs having low variable bounds any
weaker.

Assume that the total number of global variables and heap allocation statements
in a program is $Q$. We change the PCT algorithm to implement variable bounding as follows:
\begin{enumerate}
\item[0.] Choose a set of $v$
variables $q_1$,\ldots,$q_v$ representing
the minimal set of variables involved in the bug being searched ($v < d$).
\item[1.] Assign $n$ priority values $d$, $d+1$,\ldots, $d+n$ randomly to the $n$ threads.
\item[2.] Let $k_{q_1}$,\ldots,$k_{q_v}$ denote
upper-bounds on the number of instructions accessing $q_1$,\ldots$q_{v}$
respectively in any run of the program. Hence, variable $q_r$ is accessed
at most $k_{q_r}$ times in any execution of the program.
Construct a
set $S$ of elements
of the form $(q_r,j)$, where $r$ is in the
range $\lbrack1,v\rbrack$ and $j$ is in the range $\lbrack1,k_{q_r}\rbrack$.
The set $S$ will have $k=\sum_{r=1\ldots v}{k_{q_r}}$ elements.
Pick
$d - 1$ random elements from $S$ to represent the priority change
points.
\item[3.] Schedule the threads by honoring their priorities. For $i$-th chosen
element $(q_{r_i},j_i)$ in the previous step, force a
priority change point at the $j_i$th access of the $q_{r_i}$th variable. i.e., change the
priority of the thread at this point to $i$.
\end{enumerate}

We call this modified algorithm PCTVB (PCT with Variable Bounding).
Unlike PCT, where the priority change points $k_i$s are chosen randomly
from $1,\ldots,k$, PCTVB first picks a set of variables (or heap-allocation
statements), and then chooses priority change points among the accesses
to this set. PCTVB has two advantages over PCT:
\begin{enumerate}
\item As we show below, PCTVB improves the probability bound on
finding a bug with depth $d$ and variable-bound $v$. Because most bugs have
low $v$, this results in overall improvement in the bug-finding probability.
\item Choosing a set of variables apriori allows us to instrument only the program
points that can potentially access these variables. These
program points are identified using static (imprecise) alias analysis.
This is a significant
improvement over PCT where all variable accesses need to be instrumented.
\end{enumerate}

\subsection*{Probabilistic Guarantees with Variable Bounding}
The analysis of the probabilistic guarantees
of PCTVB is identical to that of PCT, as presented in the original paper~\cite{pct} and
we omit it for brevity. We simply revisit Theorem 9 (without proof) of the
original paper with our new variable bounding enhancement.
\begin{theorem}
  \label{theorem:pct_vb}
  Let $P$ be a program with
  a bug $B$ of depth $d$ and $q_1$,\ldots,$q_v$ be the minimal set of
  unique variables, accesses to which need to be preempted to trigger $B$.
  For a variable $q_i$, let $k_{q_i} \geq maxaccesses(P,q_i)$.
  Assuming $n \geq maxthreads(P)$,
$$
Pr[\text{PCTVB}(n,k,d,q_1,\ldots,q_v) \in B] \geq \frac{1}{n(\sum_{i=1..v}{k_{q_i}})^{d-1}}
$$
\end{theorem}
Here, $B$ is the set of schedules that expose the $d$-depth bug in the program.
$maxaccesses(P,q_i)$ returns the maximum number of accesses made by $P$ to
variable $q_i$ in any single run. $maxthreads(P)$ is the maximum number of threads
spawned in $P$.
The proof is identical to that of the original theorem, and is obtained by
simply replacing $k$ with $\sum_{i=1..v}{k_{q_i}}$.

For a program with $Q$ total global variables and heap allocation
statements, the probability that we pick the
correct $v$ variables ($q_1$,\ldots,$q_v$) to trigger the $v$-variable bug (if it
exists) is $\frac{1}{\binom{Q}{v}}$. Hence, the probability of
finding the $v$-variable bug is
$$
Pr[\text{PCTVB}(n,k,d,v) \in B] \geq
\frac{1}{\binom{Q}{v}}\frac{1}{n(\sum_{i=1..v}{k_{q_i}})^{d-1}}
$$

This expression depends on the sum of the access
frequencies $k_{q_1}$,\ldots,$k_{q_v}$ of
the variables $q_1$,\ldots,$q_v$. Given that the total number of variables
is $Q$, and the total number of variable accesses in any single
run is at most $k=\sum_{q_i}{k_{q_i}}$, we expect this sum to be less than $\frac{vk}{Q}$
on average (averaged over all $v$-sized sets of global variables and heap
allocation statements).
Let us assume that the sum is $f\frac{vk}{Q}$ where $f \leq 1$ on average but
could be higher depending on the set of chosen variables.
Upper-bounding $\binom{Q}{v}$
with $O(Q^v)$ for small values of $v$, this expression evaluates to
$$
Pr[\text{PCTVB}(n,k,d,v) \in B] \geq \frac{Q^{d-v-1}}{n(kvf)^{d-1}}
$$

Comparing this with PCT's original bound of $\frac{1}{nk^{d-1}}$, variable bounding
helps if
$$
Q^{d-v-1} \geq (vf)^{d-1}
$$
Assuming $v \ll Q$, variable bounding significantly improves the
lower bound on probability if $v < d-1$ and $f$ is small. In other words,
variable bounding helps if the bug being searched involves fewer variables
than its bug depth, {\em and} these variables are accessed less than
average access frequencies.

A case of particular interest are bugs with variable bound $v=1$, as
they are by-far the most common.
The inequality shows that the probability of finding
a $1$-variable bug of depth $2$ or higher improves significantly if $f < 1$.
In other words, the probability of finding bugs involving ``corner variables''
(variables used rarely compared to others) improves with variable
bounding. Intuitively, variable bounding gives all variables an equal chance,
while plain depth-bounding (or context-bounding) gives higher chance to more
frequently-accessed variables. We confirmed this experimentally by writing a small
program with two variables and varied the relative access frequencies of the variables. One of the two variables was involved in a $c=1,v=1,t=2$ concurrency
bug.
Figure~\ref{fig:pctf} shows that as the frequency of access to the variable
containing the bug is decreased, PCTVB requires
fewer executions
to find the bug compared to PCT.
\begin{figure}[htb]
  \epsfig{figure=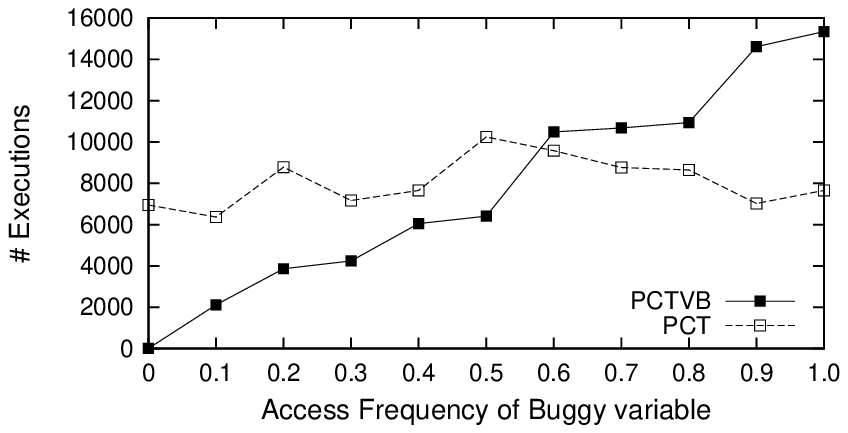, width=\columnwidth}
  \caption{\label{fig:pctf} Figure represents the number of executions required (on average) to
  trigger the bug for PCT and PCTVB as the access frequency of the buggy variable is changed.}
\end{figure}

\ifthenelse{\boolean{ThisIsTR}}{
We profiled the access frequencies of variables in different programs
in Figure~\ref{fig:variable_frequencies}. The details of these programs are
given in Table~\ref{tab:benchmarks}.}{We profiled the access frequencies of
variables of the programs listed in Table~\ref{tab:benchmarks}. Detailed graphs
can be found in our technical report \cite{tr}.} The number of accesses
varies widely across different variables for almost all benchmarks. Typically,
we expect variables with fewer accesses to undergo relatively less
testing and thus have higher likelihood of having bugs. Even if we assume
that all variables are equally likely to contain bugs, we see that
variable bounding improves the overall probability of finding a bug (if one exists).
We present a simple example.
Consider a program with a $v=1, d=2$ bug that manifests if a certain
priority sequence is followed and priority change
point occurs on a certain access $a_{q_b}$ to variable $q_b$. Assume there are
$Q$ different variables in the program, and each variable $q_i$ is accessed
at most $k_{q_i}$ number of times in any one run of the program.
Hence the probability of uncovering the bug is the probability that we pick
the correct priority sequence, and the probability that we choose $a_{q_b}$ as
the lone priority change point. The former is independent of variable bounding.
Below, we compare the latter, with and without variable bounding.

Without variable bounding, the probability of picking $a_{q_b}$ as a priority change
point is at least
$\frac{1}{\sum_{q_i}{k_{q_i}}}$ (let's call this expression $E1$). This
expression is simply the probability of choosing $a_{q_b}$ among
$\sum_{q_i}{k_{q_i}}$ potential priority change points. Notice that
$E1$ is independent of $k_{q_b}$.

With variable bounding, we first choose a variable and then choose
an access point of that variable. Hence, the probability that we pick
$a_{q_b}$ as a priority change point is $\geq \frac{1}{Q}.\frac{1}{k_{q_b}}$
(the probability that we pick $q_b$ multiplied by the probability that we pick $a_{q_b}$).
This expression depends on $q_b$ and $k_{q_b}$.
Assuming each variable is equally likely to contain a bug, further computing
the expected value of this expression over all $q_b$, we get
$\frac{1}{Q}\sum_{q_b}{\frac{1}{Qk_{q_b}}}$ (let's call this expression $E2$).

Comparing $E1$ and $E2$, and using Jensen's inequality, we get
$$
\frac{1}{\sum_{q_i}{k_{q_i}}} \leq \frac{1}{Q}\sum_{q_i}{\frac{1}{Qk_{q_i}}}
$$
or $E1 \leq E2$ with the equality happening only at
$k_{q_0}=k_{q_1}=$\ldots$=k_{q_Q}$. For typical access patterns to
variables in common programs \ifthenelse{\boolean{ThisIsTR}}{(see Figure~\ref{fig:variable_frequencies})}{\cite{tr}}, $E2$ is
expected to be significantly higher than $E1$. Hence,
assuming all variables are equally likely to have a bug, variable bounding provides
a tighter bound ($E2$) on the probability of finding the bug at $v=1, d=2$. A similar
argument holds for higher values of $v$ and $d$, and we omit the discussion for
brevity.

\ifthenelse{\boolean{ThisIsTR}}{
\begin{figure*}[htb]
\centerline{
\epsfig{figure=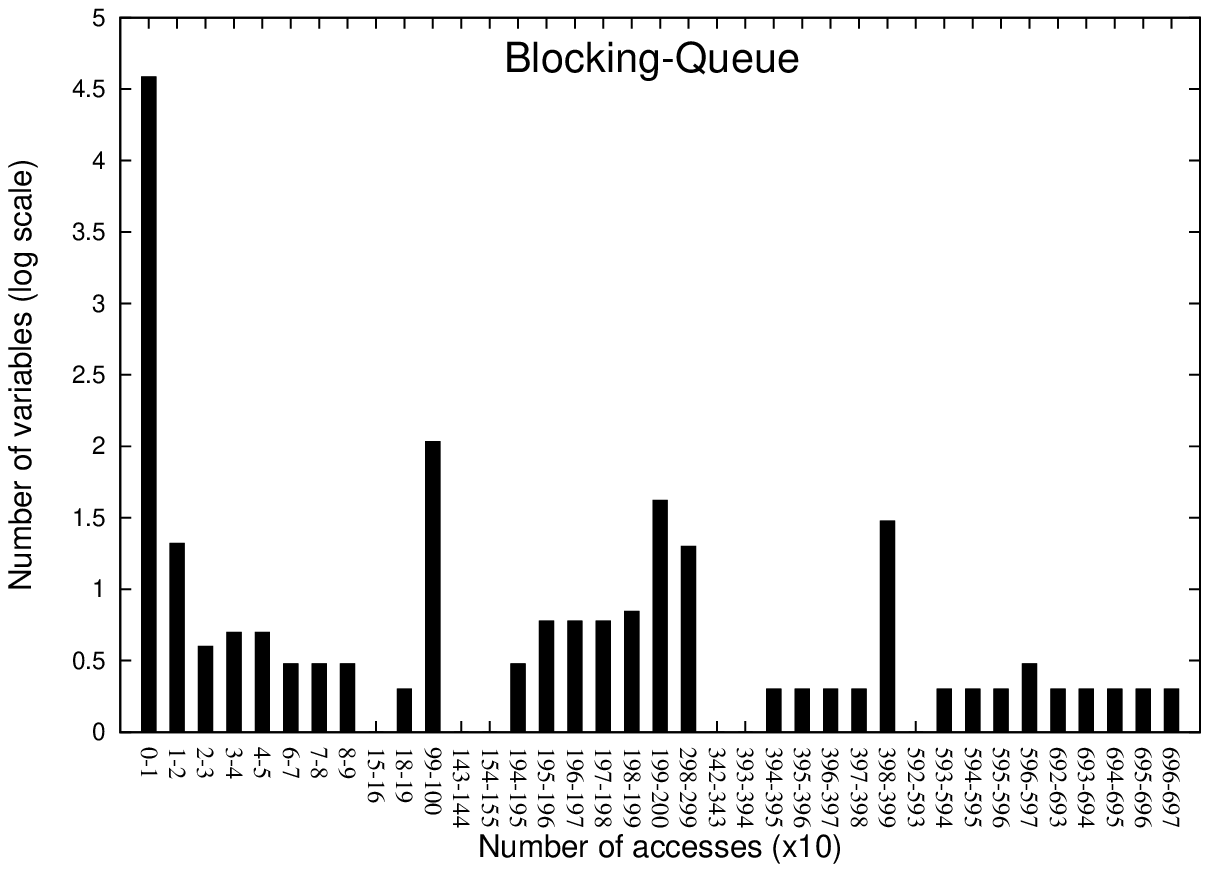, width=\columnwidth}
\epsfig{figure=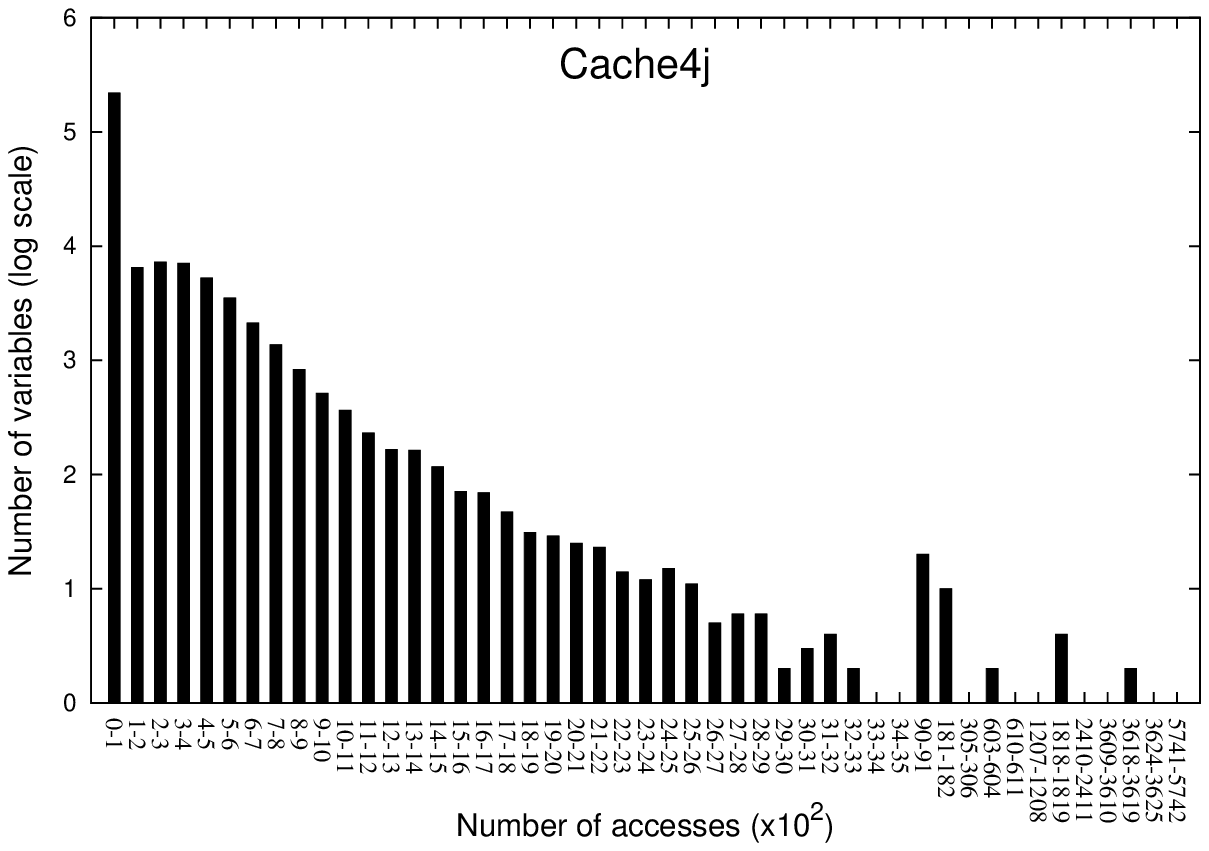, width=\columnwidth}}
\centerline{
\epsfig{figure=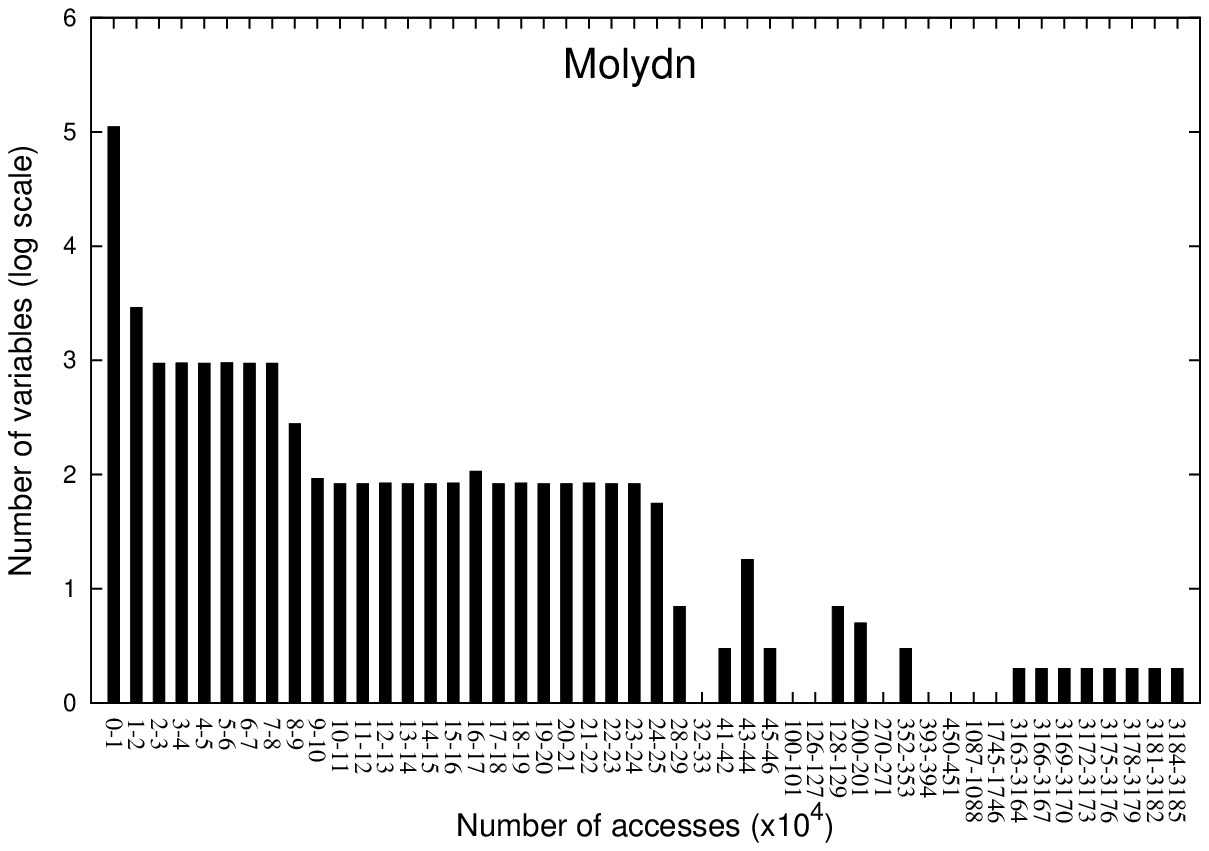, width=\columnwidth}
\epsfig{figure=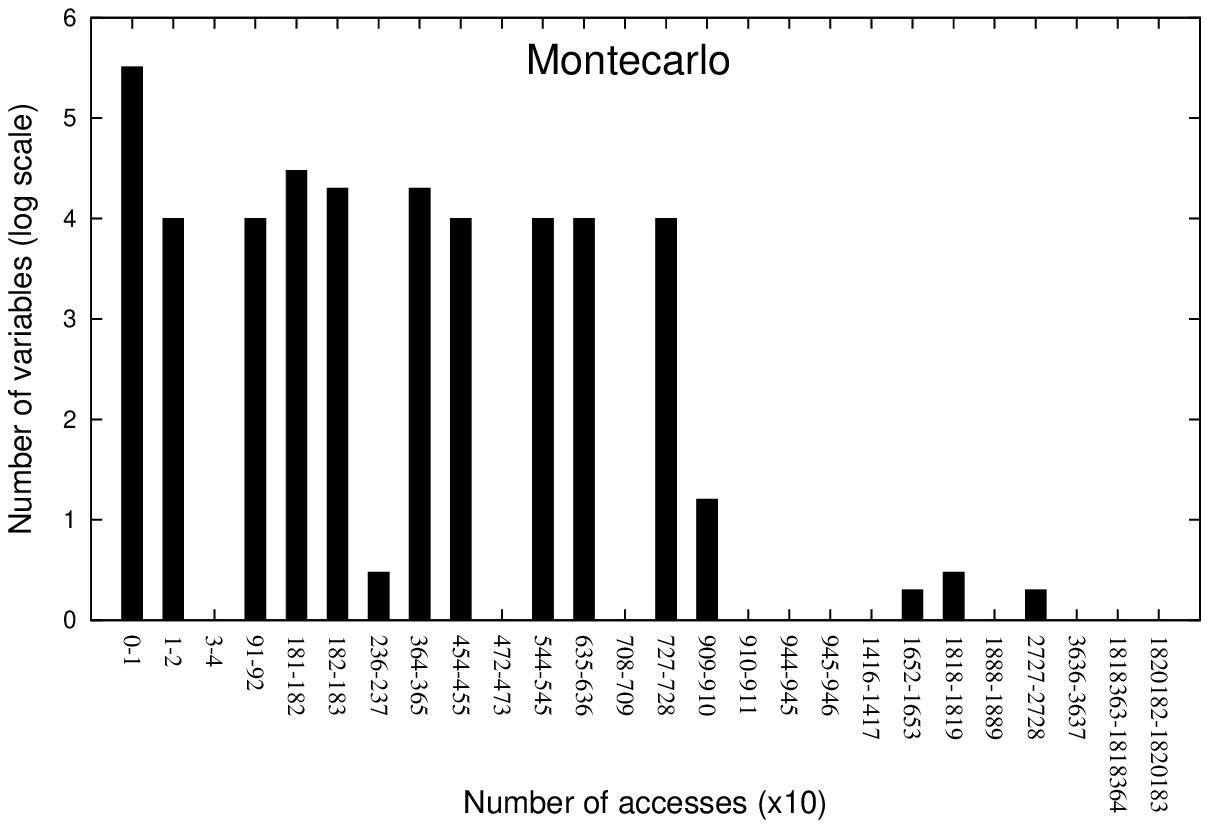, width=\columnwidth}}
\centerline{
\epsfig{figure=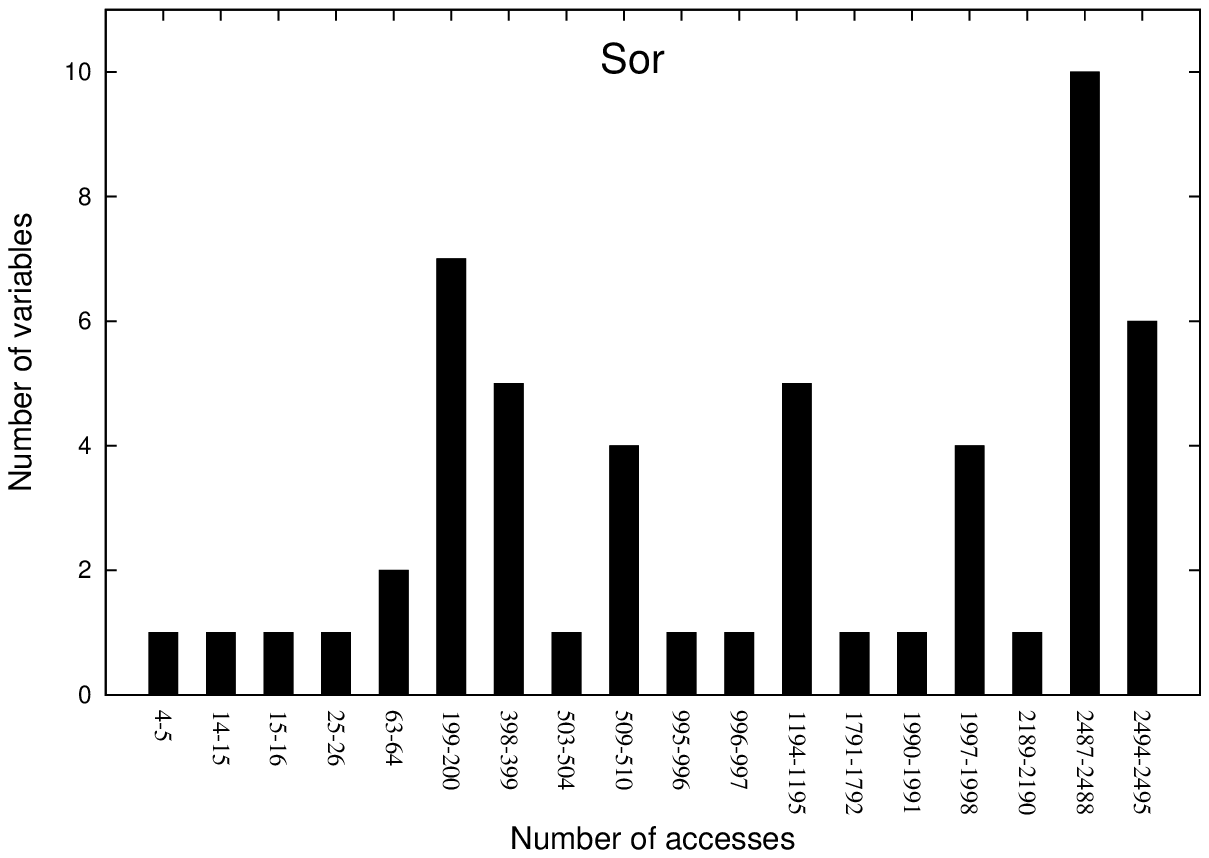, width=\columnwidth}
\epsfig{figure=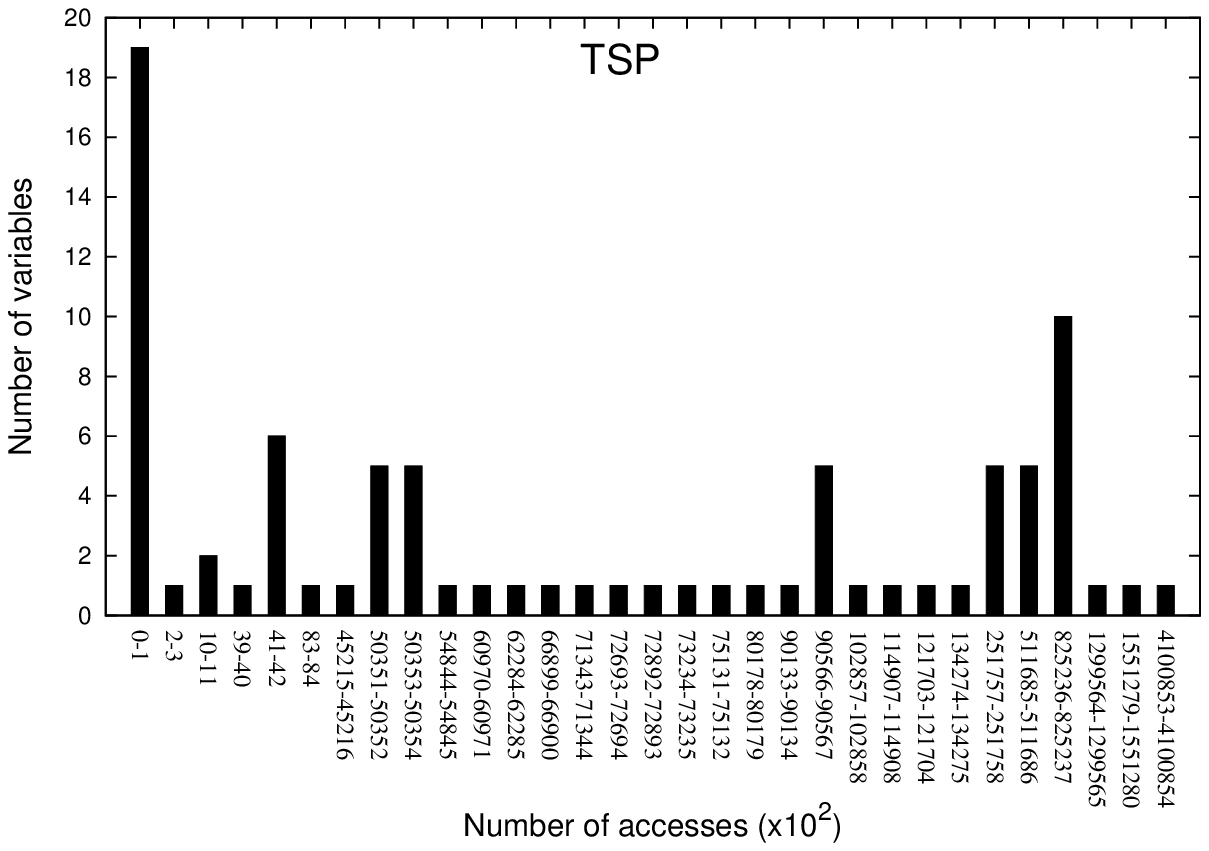, width=\columnwidth}
}
\caption{\label{fig:variable_frequencies}This figure plots the variable access
frequency profile for six of our benchmarks. The
values on the x-axis represent the frequency of access of a variable, and the y-axis plots
the number of variables that are accessed at that frequency. For example, in {\tt tsp},
19 variables are accessed between 0 to 100 times (first vertical bar), and only
1 variable is accessed between 200 to 300 times. These access frequencies were determined after running
our benchmarks multiple times on different inputs and averaging the results.}
\end{figure*}
}{}

\section{Thread Bounding}
\label{sec:threadbound}
Previous work on studying concurrency bugs found that most concurrency bugs can be
discovered by enforcing ordering constraints between a small number (typically two) of
threads~\cite{lu:learning_from_mistakes:asplos08}. This is our inspiration for
using thread-bounding while searching for concurrency bugs.
We call a bug that requires ordering constraints between at-least
$t$ distinct threads to be uncovered, a $t$-thread bug. $t$ is also called the
{\em thread-bound} of the bug.
By definition, the
thread-bound of a concurrency bug is always $2$ or higher.
Notice that our definition of thread-bound
also counts the threads that should {\em not} be executed
for a bug to manifest. For example, a bug that manifests only if thread A is
executed after thread B {\em and} thread C is not executed in between, will
be called a $3$-thread bug, and not a $2$-thread bug.
Also, $t$ is independent of $c$ and $v$. i.e., a $c$ context-bound bug and
a $v$ variable-bound bug, can have any thread bound $t \geq 2$.
Figures~\ref{fig:e-c0v0t3},~\ref{fig:e-c1v1t3},~\ref{fig:e-c2v2t3} show short
programs with $(c=0,v=0,t=3)$, $(c=1,v=1,t=3)$, and $(c=2,v=2,t=3)$ bugs, respectively.
\begin{figure}[htb]
\begin{center}
\begin{tabular}{p{2cm}|@{\ \ \ \ }p{3cm}|@{\ \ \ \ }p{3cm}}
\multicolumn{3}{c}{a = 0}\\
Thread 1:            & Thread 2:           & Thread 3:\\
\ \ a++;             & \ \ a++;            & \ \ ASSERT(a!=2);\\
\end{tabular}
\caption{\label{fig:e-c0v0t3}A short program with a $c=0,v=0,t=3$ bug}
\end{center}
\end{figure}

\begin{figure}[htb]
\begin{center}
\begin{tabular}{p{3cm}|@{\ \ \ \ }p{2.5cm}|@{\ \ \ \ }p{2.5cm}}
\multicolumn{3}{c}{a = 0}\\
Thread 1:            & Thread 2:           & Thread 3:\\
\ \ t1 = a;          & \ \ a++;            & \ \ a++;\\
\ \ t2 = a;          &                     &         \\
\ \ ASSERT(t1 $\leq$ t2+1);          &                     &         \\
\end{tabular}
\caption{\label{fig:e-c1v1t3}A short program with a $c=1,v=1,t=3$ bug}
\end{center}
\end{figure}

\begin{figure}[htb]
\begin{center}
\begin{tabular}{p{4cm}|@{\ \ \ \ }p{1.5cm}|@{\ \ \ \ }p{1.5cm}}
\multicolumn{3}{c}{\ \ \ \ a = 0}\\
Thread 1:            & Thread 2:           & Thread 3:\\
\ \ t1 = a;          & \ \ a++;            & \ \ b++;\\
\ \ t2 = a;          &                     &         \\
\ \ t3 = b;          &                     &         \\
\ \ t4 = b;          &                     &         \\
\ \ ASSERT(t1==t2 or t3==t4);         &                     &         \\
\end{tabular}
\caption{\label{fig:e-c2v2t3}A short program with a $c=2,v=2,t=3$ bug}
\end{center}
\end{figure}

\begin{sloppypar}
We posit that the number of schedules required to uncover a $t$-thread bug
increases with $t$.
For example, for a program with $n$ threads
$T_1, \cdots, T_n$, at context-bound $c=0$, all $2$-thread bugs
can be uncovered by {\em only two}
schedules, namely $\{T_{1},T_{2},T_{3},\dots,T_{n-1},T_{n}\}$ and
$\{T_{n},T_{n-1},T_{n-2},\dots,T_{2},T_{1}\}$. This is
because for any subset of 2 threads $\{T_{i},T_{j}\}$, both orders
between $T_{i}$ and $T_{j}$ (i.e., $\{T_i,T_j\}$ and $\{T_j,T_i\}$) are covered by
these two schedules.
In other words, if we arrange the threads in an arbitrary permutation,
enumerating two orders (increasing and decreasing) are enough to
uncover all $2$-thread bugs at context bound $0$.
\end{sloppypar}

A similar argument holds for $t$-thread bugs where $t > 2$.
At $c=0$, it suffices to enumerate enough schedules to explore all
$t!$ relative orderings of all $t$-sized subsets of the $n$ threads, to uncover
a $t$-thread bug.
To do this, we require an algorithm that generates enough permutations
of $n$ numbers, such that all $t!$ permutations of all $t$-sized subsets of the
$n$ numbers are exhaustively covered.

Lemma \ref{lemma:permutations} presents a randomized algorithm
to enumerate all $t!$ permutations of {\em all} $t$-sized subsets
of $n$ numbers using less than $O((t+1)! log(n))$
permutations of $n$ numbers
with a high probability. Notice that the algorithm has only
logarithmic growth with $n$, as opposed to $n!$ growth without thread bounding.

\begin{lemma}
  \label{lemma:permutations}
  The number of independent random permutations of $n$ numbers that
  need to be generated to observe {\em all} $t!$ relative orderings
  of {\em all} ${n \choose t}$ subsets of size $t$ with probability
  at least $(1-\epsilon)$,
  is $(t+1)! (log(nt) + log(\frac{1}{\epsilon}))$.
\end{lemma}

\begin{proof}
  Let $N$ be a set of $n$ distinct elements.
  Consider a fixed
  subset $S \subset N$ of $t$ elements and let $\pi$ be some
  arbitrary permutation of $S$.
  For any random permutation $\sigma$ of $n$ elements, the probability that
  $\pi$ is a subsequence of $\sigma$ is 
  $\frac{1}{t!}$ (by argument of symmetry). Hence, the probability of
  $\pi$ {\em not} appearing in $\sigma$
  is $(1 - \frac{1}{t!})$. If we enumerate $P$ independent random
  permutations of $n$ numbers, the probability of $\pi$ not appearing in
  any of the $P$ permutations
  is $(1 - \frac{1}{t!})^P$. For a fixed permutation $\pi$, let us denote
  this probability of $\pi$ not appearing in any of the $P$
  permutations by $F_{\pi}$.

  There are ${n \choose t}$ subsets of $N$ of size $t$, each
  having $t!$ permutations. Let
  us denote this set of $t!{n \choose t}$ permutations by $\Theta$.
  The probability that {\em any} one of the permutations in $\Theta$ is not
  observed in $P$ random permutations of $n$ numbers is upper-bounded by
  the sum of individual
  probabilities $\sum\limits_{\pi \in \Theta} F_{\pi} = t!{n \choose t}F_{\pi}$.
  We require this quantity
  to be less than $\epsilon$.
  $$
  t!{n \choose t}(1 - \frac{1}{t!})^P \leq \epsilon
  $$
  Writing $P$ as $(t!M)$, and approximating $(1 - \frac{1}{t!})^{t!}$ by
  $\frac{1}{e}$,
  $$
  t!{n \choose t}(\frac{1}{e})^M \leq \epsilon
  $$
  Approximating $t!$ by $t^t$, and ${n \choose t}$ by $n^t$,
  $$
  M \geq tlog({nt}) + log(\frac{1}{\epsilon})
  $$
  Replacing $M$ with $P$,
  $$
  P \geq (t+1)!(log(nt) + log(\frac{1}{\epsilon}))
  $$
  Even if $\epsilon$ is inverse-exponential in $n$, $P$ is still
  linear in $n$.
  \qed
\end{proof}

As an example, given a maximum of $n$ threads, at $t=3$,
it suffices to enumerate $(24log(n))$ random permutations of the $n$ numbers to
observe all
$3!$ relative orderings of {\em all} ${n \choose 3}$ subsets with
high probability.
For $n = 600$, we found using
simulations that $70$, $360$ and $2000$ random permutations
were enough to generate all relative orders of
all ${n \choose 3}$ ($t=3$), ${n \choose 4}$ ($t=4$)
and ${n \choose 5}$ ($t=5$) subsets respectively,
with more than 99\% probability.

To generalize to higher context-bounds, we consider a pre-empted thread as
two distinct threads (thread fragments) in this algorithm. Hence, for context-bound
$c$ bugs on a program with at most $n$ threads, we consider $n+c$ distinct thread
fragments.
To cover all $t$-thread bugs at $c$ context-bound, it suffices if
we enumerate all $(t+c)!$ permutations of all
$(t+c)$-sized subsets of the $n+c$ thread fragments. (This is more than what is strictly
required because here we are also enumerating orderings between thread fragments
belonging to the same thread). Hence, using Lemma~\ref{lemma:permutations}, the
number of schedules that need to be executed before all
$t$-thread bugs have been tested at context bound $c$ with high probability
is $O((t+c+1)!log(n+c))$.

To summarize, the exploration algorithm works as follows. A random permutation of
$1,\ldots,(n+c)$ numbers is generated at the start of each execution run.
Let us label the generated permutation $P_1,\ldots,P_{n+c}$. The scheduler
uses strict priority scheduling using $P_1,\ldots,P_n$ as the priorities
of threads $1,\ldots,n$ respectively. On the $i$th pre-emptive context switch, the
priority of the running thread is changed to $P_{n+i}$.
If $(t+c+1)!log(n+c)$ such executions are performed, each time with a new
random permutation, we expect all $t$ thread bugs at context bound $c$ to be
covered with a high probability. (If variable bounding is also being used, then
this is repeated for each set of variables). Notice that
the algorithm is independent of $t$; we only provide
probabilistic guarantees on the absence of bugs with thread-bound less
than $t$ after a certain number of schedules have been executed.

\begin{algorithm}[htb]
\caption{\label{alg:enumeration_pseudo_code}Iterative context bounding algorithm for $t$-thread bugs}
{\bf Input}:  initial state $s_0 \in$ State.
\algsetup{
  linenosize=\small,
  linenodelimiter=
  }
\begin{algorithmic}[1]
\sffamily
  \STATE struct WorkItem \{ State $state$; Priorities $prio$; \}
  \STATE Queue$<$WorkItem$>$ $WorkQueue$;
  \STATE Queue$<$WorkItem$>$ $nextWorkQueue$;
  \STATE WorkItem $w$;
  \STATE Queue$<$Priorities$>$ $threadPrios$;
  \STATE $threadPrios$.init($t$);
  \STATE int $currBound$:= 0;

  \ 

  \FOR {$prio \in threadPrios$}
  \STATE $workQueue$.Add(WorkItem ($s_0$, $prio$));
  \ENDFOR
  \WHILE {$true$}
  \WHILE {$\lnot workQueue$.Empty()}
  \STATE $w$ := $workQueue$.PopFront();
  \STATE Search($w$);
  \ENDWHILE
  \IF {$nextWorkQueue$.Empty() $||$ $currBound == c$}
  \STATE Exit();
  \ENDIF
  \STATE $currBound$ := $currBound$ + 1;
  \STATE $workQueue$ := $nextWorkQueue$;
  \STATE $nextWorkQueue$.Clear();
  \ENDWHILE

  \ 

  \STATE {\bf function} Search(WorkItem w) {\bf begin}
      \STATE WorkItem $x$; State $s$;
      \STATE TID $\textit{effTid}$;
      \STATE bool $tidenabled$, $varaccess$;
      \STATE \textbf{if} w has no successors \textbf{then return};
      \STATE Thread $tid$ := highestPriorityEnabledThread($w$.$prio$);
      \STATE $s$ := $w$.$state$.Execute($tid$);
      \STATE $tidenabled$ := {\tt ($tid \in enabled$($s$))};
      \STATE $varaccess$ := {\tt ($tid$ returned due to varaccess())};
      \STATE $x$ := WorkItem($s$, $w$.$prio$);
      \STATE Search($x$);
      \IF {($tidenabled$ $\&\&$ $varaccess$)}
      \STATE {\tt // pre-emptive cswitch. gen a schedule}
      \STATE $\textit{effTid}$ := effTidOfCurrentThread();
      \STATE changeEffTidOfCurrentThread($\textit{effTid}$+$MaxThreads$);
      \STATE $x$ := WorkItem($s$, $prio$);
      \STATE $nextWorkQueue$.Push($x$);
      \ENDIF
      \STATE \algorithmicend

\normalfont
\end{algorithmic}
\end{algorithm}

\section{Implementation}
\label{sec:implementation}
We implement variable and thread bounding in a concurrency testing tool
for Java, called RankChecker.
RankChecker
instruments the binary class code of a Java program
and associated libraries to insert appropriate schedule points. It does
not require any source-level annotations.
We instrument Java bytecode using the Javassist
library~\cite{chiba:javassist:oopsla98}.
The instrumented
test program is linked with a RankChecker library that implements a scheduler
to dictate the thread interleavings. We
implement static alias analysis using BDDs, similar to that used
in \cite{naik:chord:pldi06, whaley:alias:pldi04}.
Like previous approaches on systematic and probabilistic
testing~\cite{chess, pct}, the program under test is required
to be terminating, so that it can be run repeatedly to explore
different schedules. It is usually straightforward to convert a non-terminating
program to a terminating program.

We implement two different algorithms: exhaustive and randomized. The
exhaustive algorithm searches the state space of all schedules systematically.
The randomized algorithm searches the state space randomly, with probabilistic
guarantees on the probability of finding a bug of certain type (e.g., depth).

We first discuss the implementation of the exhaustive search strategy.
The pseudo-code is shown in Algorithm~\ref{alg:enumeration_pseudo_code}.
The algorithm is invoked for each set of variables (determined using
variable bounding). For each set of variables, a set of thread priority
orders $threadPrios$ are generated and executed. Strict
priority scheduling is followed (line~28) and priorities are
changed at variable accesses using the thread bounding algorithm (line~37).

A program state $s$ is identified by the partial thread schedule
that was executed. We implement a simple record-replay mechanism, whereby a
thread schedule is recorded and later replayed to reconstruct the same
state. As noted in \cite{chess}, replays may not result in identical states
due to other sources of non-determinism (e.g., environment, non-deterministic
calls, etc.). Our current implementation deals with these issues by enforcing
a deterministic input at all these non-deterministic points through
bytecode instrumentation.

We instrument the target program separately for each subset of
variables being tracked. For a fixed $(v,t)$ value, the enumeration
algorithm iteratively explores the
schedules with context bound $0,1,\dots,c$ ($c$ is the maximum desired
context-bound value). While enumerating schedules
for context bound $currBound$, schedules are generated for
context-bound $currBound+1$.
Our algorithm is very similar to that presented in
\cite{musuvathi:icb:pldi07}, with the following differences:
\begin{enumerate}
  \item The instrumented program points
    include memory accesses to the variables being tracked, and not just
    explicit synchronization points. As we show later, variable-bounding allows
    us to do this without significant increase in running times. Each instrumented
    program point yields control to our scheduler.
  \item When a thread yields control to the scheduler ({\tt line 29}), the
    address of the currently accessed variable is compared with the set
    of variables being
    tracked (variable bounding). Recall that it is possible that even
    though the variable access is instrumented, the accessed variable
    does not belong to the set of variables being tracked. This can happen
    either due to the imprecision of the static alias analysis or in
    cases where multiple variables are allocated by the same heap-allocation
    statement. If the accessed variable belongs to the set of variables being
    tracked, the priority of the executed thread is re-assigned, as
    discussed in Section~\ref{sec:threadbound}.
\end{enumerate}

\begin{sloppypar}
We also instrument all entries and exits from {\tt synchronized} blocks, calls
to {\tt wait} and {\tt notify}, and other thread library functions
like {\tt Thread.create}, {\tt Thread.join},
{\tt Thread.yield}, {\tt Thread.suspend} and {\tt Thread.resume}. We replace
all synchronization function calls with
calls to the appropriate scheduler functions, through instrumentation. The scheduler
function
emulates the requested operation and returns to the enumeration
algorithm (at {\tt line 29}). The enumeration algorithm then selects
the highest-priority active thread (which could have changed due to the
synchronization operation) and executes it.
For illustration, Figure~\ref{fig:wait_notify_functions} shows the scheduler's emulation
functions for {\tt wait()} and {\tt notify()}. All calls to {\tt wait()}
and {\tt notify()} in the target program are
replaced with calls to {\tt wait\_s()} and {\tt notify\_s()} respectively.
\end{sloppypar}

\begin{figure}[htb]
{\tt
  void wait\_s(cond, mutex) \{\\
  \begin{tabular}{@{}p{0.3cm}@{}l}
  & curthread.waitingOn = cond;\\
  & curthread.status = BLOCKED;\\
  & add\_to\_blocked\_threads(curthread);\\
  & wakeup\_threads\_blocked\_on(mutex);\\
  & {\em return to scheduler}\\
  \end{tabular}\\
  \}\\\\
  void signal\_s(cond, mutex) \{\\
  \begin{tabular}{@{}p{0.3cm}@{}l}
  & wakeup\_threads\_blocked\_on(cond);\\
  & {\em return to scheduler}\\
  \end{tabular}\\
  \}\\
}
\caption{\label{fig:wait_notify_functions}The scheduler's {\tt wait()} and
{\tt notify()} functions}
\end{figure}

All program instructions, where one of the
variables being tracked is
accessed, are also instrumented with a call to scheduler function
{\tt varaccess()}. The
{\tt varaccess()} function
simply returns to the enumeration algorithm (at {\tt line 29}).
The instructions that could potentially access a tracked variable are
identified using static alias analysis.

Here, we also point out that our definition of {\em context-bound} differs
from previous work~\cite{musuvathi:icb:pldi07} in a subtle way. While
the previous work counts all pre-emptive context switches towards the
context-bound, we only count the pre-emptive context switches that
violate the current priority order.
For example, in our scheme, it is possible for a low-priority thread
to be pre-empted in
favor of a high-priority thread after thread
creation, even
at $c=0$. We do not count such pre-emptions towards the context-bound.

Usually, priority-based schemes suffer from issues like
priority inversion and starvation. Because we require all our threads
to be terminating, this is not an issue in our implementation.
A priority-based scheme also violates any assumptions
of {\em strong fairness}~\cite{apt:fairness:popl87} which says that
every thread will eventually be run.
As also noted in \cite{chess}, many programs implicitly
make this assumption.
For example, while-flags (or spin-loops) are a common synchronization
construct that assume strong fairness. These loops will never terminate
if the thread that is supposed to set the condition of the loop starves.
CHESS avoids this situation by assuming that a thread yields when it is not
able to make progress, and assigning lower priority to threads calling
{\tt thread\_yield()}. In our enumeration scheme, lowering the priority of
a thread on a call to {\tt yield()} may cause certain schedules to never
get enumerated, because unlike CHESS, we enumerate only a small set
of priority orders among threads (thread-bounding). To guard against
the possibility of infinite loops, we lower the priority of a thread
if we observe that thread to {\tt yield()} more than a 100 times. This
threshold avoids infinite loops, and yet is reasonably large to not cause
interference with our thread-bounding algorithm.

Similar
to CHESS~\cite{chess}, we use
happens-before relations to construct
a happens-before graph to prune the schedules.
The happens-before graph
characterizes the partial order of related operations in a program
execution. The nodes of the happens before graph are the
executed instructions. A happens-before directed edge is drawn between two
instructions iff the two instructions execute in different threads, the
first instruction executes before the second instruction in the
given schedule, and the two instructions access the same variable of
which at-least one access is a write. The pruning is based on the
observation that two schedules with identical
happens-before graphs result in the same program state.
For a given variable
set, if one schedule has an
identical happens-before graph to another previously enumerated
schedule, this schedule (and all its derivative schedules) need not be
enumerated. Pruning is not performed across distinct variable sets.
We note that because our thread-bounding algorithm is randomized, our
exhaustive search algorithm is not strictly exhaustive. But as stated
in Lemma~\ref{lemma:permutations}, the probability
that we have not exhausted the search space can be made arbitrarily small
by executing a sufficiently large number of random priority orders.

\begin{sloppypar}
We also implement a randomized testing algorithm in RankChecker to test
variable and thread bounding. The randomized algorithm simply picks a set
of $v$ variables (globals and heap-allocation statements) randomly, and
then picks priority change points at accesses to these variables. The
values of the maximum number of accesses, $k_{q_1}$,\ldots,$k_{q_Q}$, to
variables, $q_1$,\ldots,$q_Q$ respectively, are estimated by
running the program without priority scheduling multiple times
and counting the average number of accesses to each variable in these runs. The
priority change points are picked uniformly over the interval $[1,k_{q_i}]$.
Our randomized algorithm is modeled after PCT's depth-bounding. The only difference
between our algorithm and PCT is in the assignment of priorities. PCT generates
a set of random priority orders, such that each thread gets to be the lowest
priority thread in at least one of the priority orders. Also, on a priority change
point, PCT decreases the priority of the current thread to become lower than
the priority of all currently-executing threads. Our priority orders are instead
chosen using the thread-bounding algorithm given in Section~\ref{sec:threadbound}.
\end{sloppypar}

\section{Experimental Results}
\label{sec:results}
We perform experiments to answer the following questions:
\begin{itemize}
\item What are the typical values of variable-bound and thread-bound in common
concurrency bugs?
\item What is the runtime improvement due to variable bounding?
\item For exhaustive search strategy, do variable and thread bounding improve the
number of executions required to expose a bug?
\item For randomized search strategy, do variable and thread bounding improve the
number of executions required to expose a bug?
\end{itemize}

We picked a variety of small and large Java programs and one C\# program as test
programs to evaluate our algorithms.
The details of these programs are given in
Table~\ref{tab:benchmarks}. The first 13 programs are from the ConTest Concurrency
Benchmark Suite \cite{eytani:concbench:jconcom07}. All these programs contain a concurrency bug.
The next 8 benchmarks are multi-threaded Java programs commonly used to evaluate
concurrency testing and verification tools. Some of these programs contain
bugs. The last program ({\tt RegionOwnership}) is a C\# program containing a
reasonably complex concurrency bug. This program has been previously analyzed using
CHESS \cite{EmmiQR11}.
As we discuss later, we have also implemented
variable and thread bounding in the CHESS tool to test C\# programs. We report our
experiences with variable bounding on the {\tt RegionOwnership} benchmark.
Within a variable and thread bound, we further rank our schedules
based on the loop iteration number (recall Section~\ref{sec:varbound}).
For exhaustive search experiments, while choosing our variable set, we give
priority to shared variables. i.e., variables known to be shared are chosen
before
other variables. A variable is known to be shared if in one of the preparatory runs,
we found a variable being accessed by at least two threads.

We ran RankChecker on the programs containing known bugs
with variable bounding to check the
bug characteristics. Table~\ref{tab:benchmarks} lists the $(c,v,t)$ values at
which these bugs were uncovered using the exhaustive algorithm. We found that
all these bugs were $c\leq2,v\leq2,t=2$ bugs. We also surveyed past
papers on studying concurrency
bugs and bug databases of popular applications, to study
the bugs reported in them. We found that all these bugs
were also of type $c\leq2,v\leq2,t=2$.

We provide pseudo-code of the $c=2,v=1,t=2$ bug found in {\tt AllocationVector} in
Figure~\ref{fig:allocation_vector}.
\begin{figure}[htb]
{\small
\begin{tabular}{@{}l@{}|@{}l@{}}
                                  &    {\tt \ Block b = FindFreeBlock();}\\
{\tt Block b = FindFreeBlock();}  &    {\tt \ } {\em first context switch}\\
{\tt ASSERT(IsBlockFree(b));}     &    \\
{\tt MarkBlockAllocated(b);}      &\\
{\em second context switch}       &    {\tt \ ASSERT(IsBlockFree(b));}  !FAILS!\\
                                  &    {\tt \ MarkBlockAllocated(b);}\\
                                  &    {\tt \ FreeAllBlocks();}\\
{\tt FreeAllBlocks();}            &\\
\end{tabular}
}
\caption{\label{fig:allocation_vector}Pseudo-code showing the $c=2,v=1,t=2$ bug
in AllocationVector. The routines {\tt FindFreeBlock()}, {\tt MarkBlockAllocated()},
and {\tt IsBlockFree()} are all synchronized (i.e., protected by a monitor lock).
{\tt FindFreeBlock()} searches a global vector to find an unallocated block.
{\tt MarkBlockAllocated()} sets a flag in block {\tt b} and {\tt IsBlockFree()} checks
that flag.}
\end{figure}

\begin{table*}[htb]
{
  \begin{center}
    \scriptsize
    \setlength{\tabcolsep}{2pt}
    \begin{tabular}{l|r|r|r|r|l|r|r}
      Benchmark & SLOC & \# Threads & \# Variables & Bug? & Description & Schedules & $(c,v,t)$ \\
                &      &            &              &      &             & Explored &           \\
      \hline
      \hline
      \multicolumn{6}{l}{ ConTest Benchmarks }\\
      \hline
      MergeSort & 376 & 100 & 564 & Yes & Sorts a set of integers using mergesort& 651 & $(1,1,2)$\\
      \hline
      Producer Consumer & 279 & 7 & 61 & Yes & Simulates producer-consumer behavior& 1 & $(0,0,2)$\\
      \hline
      LinkedList & 420 & 3 & 60 & Yes & LinkedList's implementation with test-harness& 23 & $(1,1,2)$\\
      \hline
      BubbleSort & 365 & 9 & 54 & Yes & Sorts a set of integers using bubblesort& 1 & $(0,0,2)$\\
      \hline
      BubbleSort2 & 129 & 101 & 105 & Yes & Sorts a set of integers using bubblesort& 2 & $(0,0,2)$\\
      \hline
      Piper & 210 & 9 & 33 & Yes & Manages airline reservations& 64 & $(1,1,2)$\\
      \hline
      Allocation Vector & 288 & 3 & 4010 & Yes & Manages free and allocated blocks in a vector& 113 & $(2,1,2)$\\
      \hline
      BufWriter & 259 & 5 & 27 & Yes & Reads and writes to a buffer concurrently& 12 & $(0,0,2)$\\
      \hline
      PingPong & 276 & 18 & 25 & Yes & Simulates the behavior of ping-pong game& 234 & $(1,1,2)$\\
      \hline
      Manager & 190 & 6 & 25 & Yes & Manages free and allocated blocks& 33 & $(1,1,2)$\\
      \hline
      MergeSortBug & 258 & 29 & 52 & Yes & Sorts a set of integers using mergesort& 89 & $(1,1,2)$\\
      \hline
      Account & 169 & 3 & 26 & Yes & Manages a bank account& 19 & $(1,1,2)$\\
      \hline
      AirLineTickets & 99 & 11 & 5 & Yes & Simulates selling of airline tickets& 2 & $(0,0,2)$\\
      \hline
      \hline
      \multicolumn{6}{l}{ Java's Library in JDK 1.4.2}\\
      \hline
      HashSet & 7086 & 200 & 4777 & Yes & Thread-safe implementation of HashSet& 813 & $(1,1,2)$\\ 
      \hline
      TreeSet & 7532 & 200 & 6140 & Yes & Thread-safe implementation of TreeSet& 813 & $(1,1,2)$\\
      \hline
      \hline
      \multicolumn{6}{l}{ Other Java Benchmarks}\\      
      \hline
      Cache4j & 3897 & 12 & 251,469 & No & Cache implementation for Java objects& - & - \\
      \hline
      Molydn & 1410 & 8 & 121,371 & No & Benchmark from Java Grande Forum& - & -  \\
      \hline
      Montecarlo & 3630 & 8 & 452,700 & No & Benchmark from Java Grande Forum& - & - \\
      \hline
      TSP & 719 & 18 & 84 & No & Travelling Sales Problem's implementation& - & - \\
      \hline
      Blocking Queue & 57 & 3 & 38,828 & No & Tests BlockingQueue library implementation& - & - \\
      \hline
      Sor & 17,738 & 6 & 53 & No & Successive Order Relaxation method's implementation& - & - \\
      \hline
      \hline
      \multicolumn{6}{l}{ C\# Benchmark}\\
      \hline
      RegionOwnership & 1500 & 5 & 41 & Yes & Manages coordination for objects & 47248 & $(2,2,2)$\\
      & & & & & communicating using async calls & & \\
      \hline
      
    \end{tabular}
  \end{center}
\caption{\label{tab:benchmarks}Test programs and their details. The last two columns list, for each buggy program,
the number of schedules explored until we found the first bug and tuple $(c,v,t)$ at which the bug occurs.}
}
\end{table*}

We next discuss the improvements in running time due to variable bounding.
Table~\ref{tab:pref_results} shows our results on some of our Java programs. The
other Java programs were too small to show any meaningful improvements. The runtime
statistics have been averaged over several runs of the programs.
With variable bounding, there is up to 100x improvement in the runtime cost of
instrumentation. The runtime improvement depends on the proportion
of computation and I/O
in the test program. Variable bounding results in
improvement because only program statements
identified
by alias analysis as potential accesses to our set of tracked variables need
to be
instrumented. The performance of an instrumented run is now comparable to that
of a native run, which makes it practical to implement systematic
testing algorithms where all variables are considered
as potential pre-emptions points. (The native run is
sometimes slower than the instrumented run; this happens due to the
overhead of process creation in the native run which does not
exist in our instrumented run.). This is a significant improvement over
previous work, where only synchronization operations have been considered as
potential pre-emption points~\cite{chess,pct}.
\begin{table*}[htb]
\begin{center}
\scriptsize
\begin{tabular}{l|r|r|r|r|r|r|r|r|r}
	Program Name & BCI & var sites & \# of accesses & Native time(sec) & v$0$(sec) & v$1$(sec) & v$2$(sec) & v-$all$(sec) & v-$all$/v$1$ \\
    \hline
    \hline
    Cache4j & 231.1m & 101 & 21.4m & 0.34 & 0.47 & 1.23 & 2.76 & 26.38 & 21.3 \\
    \hline
    Molydn & 2.33b  & 209 & 1.4b & 0.39 & 3.15 & 11.59 & 19.86 & 1239.76 & 106.3 \\
    \hline
    Montecarlo & 577.7m & 235 & 446.96m & 0.48 & 1.94 & 4.74 & 5.21 & 323.12 & 68.08 \\
    \hline
    TSP & 8.76b & 65 & 2.55b & 4.2 & 4.23 & 32.64 & 109.72 & 1180.28 & 36.15 \\
    \hline
    Blocking Queue & 3.4m & 13 & 0.65m & 0.17 & 0.18 & 0.194 & 0.202& 1.386 & 7.14 \\
    \hline
    Sor & 0.2m & 46 & 0.68m & 0.07 & 0.25 & 0.249 & 0.348 & 0.392 & 1.57 \\
    \hline    
    HashSet & 157.4k & 137 & 16889 & 0.07 & 0.0775 & 0.0901 & 0.0976 & 0.2687 & 2.98 \\ 
    \hline
    TreeSet & 113k & 146 & 16273 & 0.69 & 0.078 & 0.089 & 0.09 & 0.259 & 2.91 \\
    
\end{tabular}
\end{center}
\caption{\label{tab:pref_results}The different columns in this table represents 
the name of the program, the (average) number of byte code instructions executed by the program, total number of different instrumentation sites, 
which includes heap-allocation statements and global variables, total number of accesses, native execution time, the average amount of time 
taken for one execution when we are tracking 0, 1, 2, and all variables, respectively, and the last column represents the ratio of the v-$all$ and v$1$ columns.}
\end{table*}

Previously, a tool called RaceFuzzer~\cite{sen:racefuzzer:pldi08} reported a
$c=1,v=1,t=2$ concurrency bug (data race) in {\tt cache4j}. Our tool could not find
this bug even after exhaustively enumerating all schedules up to $c\leq2,v\leq2,t=2$. On
deeper
inspection, we found that the bug did not exist.
It turned out that RaceFuzzer had generated a false bug report
due to an error in
the modelling of the semantics of the Java interrupt exception in the tool. We
reported this to the author of RaceFuzzer~\cite{sen:racefuzzer:pldi08}, and he
did not object to our findings.
Because RankChecker
actually runs a schedule to try and trigger assertion failures,
a bug report and the associated schedule reported by it also serve as a
proof of the bug's existence.

\subsection*{Variable and Thread Bounding in CHESS}

We further validate the effectiveness of variable and thread
bounding in practice by implementing
it inside CHESS\cite{chess} and testing it on C\# benchmarks that were
previously
used with CHESS \cite{EmmiQR11}. However, we did not have an alias analysis
readily available for
C\#, thus, we only implemented a simple form of variable bounding that works as follows.
Let VT-CHESS refer to our extension of CHESS with variable bounding.
Suppose VT-CHESS is executed on program $P$ with variable bound $v$ and pre-emption bound $c$. If $v \geq c$
then VT-CHESS behaves exactly like CHESS. When $v < c$, then during an execution of
$P$, VT-CHESS
records the shared variables accessed just before the first $v$ pre-emptions in the
execution. Subsequent
pre-emptions ($v+1^\text{th}$ to $c^\text{th}$) are constrained to occur only after an access of one of 
these $v$ variables. In other words,
the $v$ variables for variable bounding are chosen dynamically.

The deepest reported bug found using CHESS is in a program called {\tt RegionOwnership}. It is a C\# library
that manages concurrency and coordination for objects communicating via asynchronous procedure
calls. The library is accompanied by a single test case comprising of a
one-producer one-consumer system. The library is $1500$ lines of code, and an execution access a synchronization
variable at most $280$ times. The test reveals a bug that requires at least $3$ pre-emptions. 

Table~\ref{tab:vb-chess} shows the number of executions and time taken before VT-CHESS either reported a bug
or finished exploring all behaviors under the given bound. We used $c=3,t=2$ in all invocations of VT-CHESS.
VT-CHESS was able to find the bug about $6$ times 
faster then CHESS while using a variable bound of $2$. Using a variable bound of $1$ does not expose the bug, 
but Table~\ref{tab:vb-chess} shows a further reduction in search space when this bound is imposed.

\begin{table}[htb]
{\scriptsize
  \begin{center}
  \setlength{\tabcolsep}{2pt}
  \begin{tabular}{|l|r|r|r|r|}
    \hline
     & Bug found? & \# Executions & Time (sec) \\
    \hline
    No VB,$t=2$ & Yes & 132507 & 6897.3 \\
    $v=2$,$t=2$  & Yes & 47248  & 1224.4 \\
    $v=1$,$t=2$  & No & 30437  & 581.0 \\
    \hline
  \end{tabular}
\end{center}
\caption{\small \label{tab:vb-chess}Experiments with the {\tt RegionOwnership} benchmark.}
}
\end{table}

\subsection*{Variable and Thread Bounding in Randomized Algorithms}
All the bugs found in our test programs, except {\tt RegionOwnership}, were of type $v\leq1$.
As seen in Tables~\ref{tab:pref_results}~and~\ref{tab:vb-chess}, variable
bounding improves both runtime and the number of schedules explored while
systematically testing concurrent programs. To further study the effect
on bugs with higher $v,t$ values, we modified one of our test programs such
that it had a bug of the required type and ran RankChecker on it.
Table~\ref{tab:montecarloresults} presents our results.

As
expected, the time required to find the bug decreases dramatically with variable
bounding. The number of executions required to find a $v=0$ bug is
roughly the same with and without variable bounding, but increases with the thread-bound
of the bug.
The number of executions required to find the bug improves
with variable bounding at $v\geq1$, for the reasons discussed in
Section~\ref{sec:varbound-random}.

\begin{table}[htb]
{\scriptsize
  \begin{center}
  \setlength{\tabcolsep}{2pt}
  \begin{tabular}{|l|r|r|r|r|}
    \hline
    Bug Type & \multicolumn{2}{|c|}{With v,t bounding} & \multicolumn{2}{c|}{Without v,t bounding} \\
    $(c,v,t)$ & \# Executions & Time (sec) & \# Executions & Time (sec) \\
    \hline
    \hline
    $(0,0,2)$ & 3.1 & 10.9 &  2.7 & 768.9 \\
    \hline
    $(0,0,3)$ & 3.7 & 12.1 &  3.5 & 1010.7 \\
    \hline
    $(0,0,4)$ & 4.9 &  14.9 &  3.8 & 1101.7 \\
    \hline
    $(1,1,2)$ & 1636.2 & 6409.9 & - & TimedOut \\
    \hline
    $(1,1,3)$ & 4889 & 20371.2 & - & TimedOut \\
    \hline
    $(2,1,2)$ & 28121 & 112950.3 & - & TimedOut \\
    \hline
  \end{tabular}
\end{center}
}
\caption{\small \label{tab:montecarloresults}This table represents the average number
of executions and time required to capture a bug of type $(c,v,t)$ introduced
in the Montecarlo benchmark. Without variable bounding, our tool
timed out after executing for more than 3 days for $c\geq1,v\geq1$, without
finding the bug.}
\end{table}

\section{Related Work}
\label{sec:relwork}
There is a large body of work on
static~\cite{bessey:coverity:acmcomm10, naik:chord:pldi06} and
dynamic~\cite{savage:eraser:tcs97, elmas:goldilocks:pldi07,
joshi:deadlockfuzzer:pldi09, lu:avio:asplos06,
ctrigger} techniques to uncover
concurrency bugs. While many of these
techniques are very effective and have uncovered a variety of
previously-unknown bugs in well-tested software, the current dominant
practice in the software industry still remains {\em stress testing}.
There are a few important likely reasons for
this (apart from plain inertia):
\begin{enumerate}
  \item Plain testing is more natural.
  \item Most tools target a small class of bugs. For example, some tools
    target only dataraces, others target only atomicity-violation bugs, and
    yet others target only deadlocks. It is confusing for a developer to
    understand the function of each tool and apply them separately.
  \item Many tools have false positives. Spending time and energy on
    a false-positive bug report is annoying and counter-productive.
  \item Many tools require source-level annotations. Many tools rely on
    certain programming disciplines. e.g., all shared memory accesses
    must be protected by a lock. At places where the programmer deliberately
    violates this discipline, source code annotations are required.
    Most developers are usually reluctant to annotate their source code for
    better testing.
  \item High-Runtime, Low Coverage: Many tools have a high runtime cost, and
    provide low coverage.
\end{enumerate}

Model checking approaches
\cite{bruening:systematic_testing, chess,
godefroid:verisoft:book96, godefroid:verisoft:popl97,
stoller:modelcheck:spin00, khurshid:pathfinder:tacas03} are
closer
to the familiar idiom of testing. Our
model checker targets all types
of concurrency bugs, has no false positives, and requires
no source-level annotations.
We address the state explosion problem
by ranking the schedules using $v,t$ to
maximize coverage in the first few schedule executions.
Previous approaches have reduced this search space by either
limiting context switches only at synchronization
operations~\cite{chess} (resulting in potential false negatives), or
using an offline memory trace of the program to identify and rank
unserializable interleavings~\cite{ctrigger} (primarily to identify
atomicity-violation bugs). We believe that variable and thread bounding are
more general methods of ranking (or reducing) the search space.

CHESS~\cite{chess} uses iterative context bounding to 
rank schedules. We borrow many ideas from CHESS, including
iterative context bounding~\cite{musuvathi:icb:pldi07}, using
a happens-before graph for stateless
model checking, and fair scheduling. We provide further ranking
of schedules to uncover most bugs with a smaller number of schedules. While
we consider all shared memory accesses as potential context-switch
points, CHESS only allows pre-emptible context switches at
explicit synchronization primitives.
This restriction (first used in ExitBlock~\cite{bruening:systematic_testing}) is
justified if
all shared memory accesses are
protected by explicit synchronization (e.g., {\tt lock/unlock}).
CHESS relies on a data-race detector to separately check this property.
Even if we assume a precise and efficient data-race detector, this approach
still overlooks ``adhoc'' synchronization that do not involve
known synchronization primitives~\cite{xiong:adhoc:osdi10}.

VeriSoft~\cite{godefroid:verisoft:book96, godefroid:verisoft:popl97, godefroid:verisoft:issta98} also uses an
exploration strategy to model-check ``distributed'' systems
using a {\em state-less search} (i.e., storage of previously-visited
states are not required). Verisoft uses partial-order
methods to reduce redundancy, similar to happens-before graph pruning
used in CHESS and in our tool. S. Stoller~\cite{stoller:modelcheck:spin00}
uses a similar approach to model-check multi-threaded distributed Java
programs. We believe that variable and thread bounding ideas are equally relevant
to these model checking approaches as well.

Our work is complementary to
race-detection tools~\cite{savage:eraser:tcs97, naik:chord:pldi06,
sen:racefuzzer:pldi08, yu:racetrack:sosp05, marino:literace:pldi09},
deadlock-detection tools~\cite{joshi:deadlockfuzzer:pldi09}, and
atomicity-violation detection tools~\cite{lu:avio:asplos06, ctrigger}.
We do not focus on a
particular
class of bugs, but rather drive a model checker into exploring
interesting schedules that are likely to uncover all these bugs early.
In practice, small $v,t$ values uncover most data races, deadlocks, and
atomicity-violation bugs.

\section{Conclusion}
\label{sec:conclusion}
We present variable and thread bounding to rank thread schedules for
systematic testing of concurrent programs. Through experiments on a
variety of Java and C\# programs, we find that the ranking significantly
aids early discovery of common bugs.

\ifthenelse{\boolean{ThisIsTR}}{
\section{Acknowledgements}
Lemma~\ref{lemma:permutations} and its proof are due to Sandeep Sen (IIT Delhi).
}{}

\bibliographystyle{abbrv}
\bibliography{varbounding}
\end{document}